\documentclass[11pt,draft]{article}

\usepackage{ifdraft}
\ifdraft{\newcommand{\authnote}[3]{{\color{#3} {\bf  #1:} #2}}}{\newcommand{\authnote}[3]{}}

\usepackage{fullpage}
\usepackage[usenames,dvipsnames,svgnames,table]{xcolor}
\usepackage[english]{babel}
\usepackage[]{hyperref}
\usepackage{amsmath,amssymb,amsthm}
\usepackage{mathtools}
\usepackage{graphicx}
\usepackage{caption}
\usepackage{enumitem}

\usepackage{tikz,pgfplots}
\pgfplotsset{compat=1.10}

\let\OLDthebibliography\thebibliography
\renewcommand\thebibliography[1]{
  \OLDthebibliography{#1}
  \setlength{\parskip}{0pt}
  \setlength{\itemsep}{0pt plus 0.3ex}
}

\usepackage{thmtools}
\usepackage{thm-restate}

\makeatletter
\newcommand{\mypm}{\mathbin{\mathpalette\@mypm\relax}}
\newcommand{\@mypm}[2]{\ooalign{%
  \raisebox{.1\height}{$#1+$}\cr
  \smash{\raisebox{-.6\height}{$#1-$}}\cr}}
\makeatother

\newcommand{\cube}{\{\pm 1\}^n}

\newcommand{\ignore}[1]{}

\newcommand{\N}{\mathbb{N}}

\newcommand{\R}{\mathbb{R}}

\newcommand{\C}{\mathbb{C}}

\DeclareMathOperator{\Tr}{Tr}

\renewcommand{\sup}{\mathrm{sup}}
\renewcommand{\min}{\mathrm{min}}

\renewcommand{\max}{\mathrm{max}}

\renewcommand{\epsilon}{\varepsilon}

\newcommand{\Diag}{\text{\rm Diag}}
\def\01{\{0,1\}}

\newcommand{\ind}[2]{\underline{#1}}

\newcommand{\cb}{_\mathrm{cb}}

\newtheorem{defin}{Definition}[section]

\newtheorem{theorem}[defin]{Theorem}

\newtheorem{lemma}[defin]{Lemma}

\newtheorem*{claim*}{Claim}

\newtheorem*{conjecture*}{Conjecture}
\theoremstyle{definition}

\DeclareMathOperator{\diag}{Diag}

\DeclareMathOperator{\cbdeg}{cb-deg_\epsilon}
\DeclareMathOperator{\bmdeg}{bm-deg_\epsilon}

\newcommand{\Hilbert}{\mathcal H}

\date{}
\begin{document}

\title{Semidefinite programming formulations for the completely bounded norm of a tensor} 
\author{Sander Gribling\thanks{Centrum Wiskunde \& Informatica (CWI) and QuSoft, Amsterdam, The Netherlands. The work was supported by the Netherlands Organization for Scientific Research, grant number 617.001.351.} 
\and Monique Laurent\thanks{Centrum Wiskunde \& Informatica (CWI) and QuSoft, Amsterdam and Tilburg University, The Netherlands.}}
\maketitle

\begin{abstract}
We show that a certain tensor norm, the \emph{completely bounded norm}, can be expressed by a semidefinite program. This tensor norm recently attracted attention in the field of quantum computing, where it was used by Arunachalam, Bri\"{e}t and Palazuelos for characterizing the quantum query complexity of Boolean functions. Combined with their results, we obtain a new characterization of the quantum query complexity through semidefinite programming. Using the duality theory of semidefinite programming we obtain a new type of certificates for large query complexity. We show that our class of certificates encompasses the linear programming certificates corresponding to the approximate degree of a Boolean function.
\end{abstract}

\section{Introduction}

Throughout, we let $T = (T_{i_1,\ldots, i_t}) \in \R^{n \times \cdots \times n}$ be a $t$-tensor acting on $\R^n$. The \emph{completely bounded norm} of $T$ is defined as 
\begin{equation} \label{eqdef}
||T||_\mathrm{cb} := \sup\Big\{\Big|\Big|\sum_{i_1,\ldots,i_{t} = 1}^n T_{i_1,\ldots,i_{t}} U_1(i_1)\cdots U_{t}(i_{t})\Big|\Big|\ :\  d \in \N, \ U_j(i) \in O(d) \text{ for } i \in [n], j \in [t]\Big\}.
\end{equation}
Here $||\cdot ||$ is the operator norm and $O(d) \subseteq \R^{d \times d}$ is the group of $d \times d$ orthogonal matrices. Note that in~\eqref{eqdef} one could equivalently optimize over \emph{complex} unitary matrices $U_j(i)$. 

We first show that $||T||_\mathrm{cb}$ can be expressed as the optimal value of a semidefinite program (SDP). 
This SDP involves matrices of size $O(n^{\lceil t/2\rceil})$ and $O(n^{2\lceil t/2\rceil})$ linear constraints, so that an additive $\epsilon$-approximation of its optimal value can be obtained in time~$\mathrm{poly}(n^t, \log(1/\epsilon))$ (see Theorem \ref{thm2} in Section \ref{secSDP}).

To put this result in perspective, if we replace the product $U_1(i_1)\cdots U_{t}(i_{t})$ by the Kronecker (or tensor) product $U_1(i_1) \otimes \cdots \otimes U_{t}(i_{t})$ then we obtain the \emph{jointly completely bounded norm} of $T$. It is known that there is a one-to-one correspondence between the jointly completely bounded norm of a $t$-tensor and the entangled bias of an associated $t$-partite XOR game. The latter can be computed in polynomial time when $t=2$~\cite{Tsir87}, but it is an NP-hard problem to give any constant factor multiplicative approximation of the entangled bias of a $3$-partite XOR game~\cite{Vid16}. 
Hence the jointly completely bounded norm of a $3$-tensor is hard to approximate up to any constant factor.

A main motivation for  our study of the completely bounded norm of a tensor comes from  a connection to the quantum query complexity of  Boolean functions that was recently shown in~\cite{ABP17}, see Section~\ref{sec:qquerycompl}. As an application of our result, the quantum query complexity of a Boolean function $f$ can be obtained by checking feasibility of some SDPs. Using the duality theory of semidefinite programming we obtain a new type of certificates for large query complexity. We show that our class of certificates encompasses the linear programming certificates corresponding to the approximate degree of $f$ and we propose an intermediate class of certificates based on second-order cone programming. Previously, two other semidefinite programming characterizations of quantum query complexity were given in~\cite{BSS03,HLS07} using a different approach. For total functions on $n$ bits, these two SDPs have matrix variables of size $2^n$ while our SDP has a matrix variable of size $\Theta(n^t)$ where $t$ is the number of queries. Thus, for small query complexity (constant) the matrix variable in our SDP is much smaller. In Section~\ref{sec:comparison} we compare the three SDPs in more detail. 

We point out that the notion of completely bounded norm of a tensor considered in the present paper differs from the notion considered in the work of Watrous~\cite{Wat09}.

\section{Preliminaries}

We first recall some basic properties of semidefinite programs. Let $\mathcal S^{N}$ denote the vector space of real symmetric $N \times N$ matrices equipped with the trace inner product $\langle A,B \rangle = \mathrm{Tr}(AB)$, and let $\mathcal S_+^{N} \subset \mathcal S^N$ be the cone of positive semidefinite $N \times N$ matrices. A set of matrices $C,A_1,\ldots, A_m \in \mathcal  S^N$ and a vector $b \in \R^m$ define a pair of  semidefinite programs, a \emph{primal} $(P)$ and a \emph{dual} $(D)$:
\begin{align*} 
(P)  \qquad \max \quad & \langle C, X \rangle &(D) \qquad  \min \quad & \langle b, y\rangle  \\ 
\text{s.t.} \quad& X \in \mathcal S^N_+ &\text{s.t.} \quad&y \in \R^m \\ 
&\mathcal A(X) = b \qquad &&\mathcal A^*(y) - C \in \mathcal S^N_+
\end{align*}
Here,  $\mathcal A:\mathcal S^{N} \to \R^m$ is the linear operator defined by $\mathcal A(X) = (\mathrm{Tr}(A_1 X), \ldots, \mathrm{Tr}(A_m X))$, whose adjoint $\mathcal A^*$ acts on $\R^m$ as $\mathcal A^* (y) = \sum_{i=1}^m y_i A_i$, so that $\langle \mathcal A(X),y\rangle=\langle X, \mathcal A^*(y)\rangle$.

We always have \emph{weak duality}: the optimal value of the primal problem is a lower bound on that of the dual. Moreover, if $(P)$ is strictly feasible (i.e., 
there exists a feasible $X$ with $X-r I \succeq 0$ for some scalar $r>0$) 
and the feasible region of $(P)$ is bounded (i.e., there exists a scalar $R>0$ such that   $\langle X,X \rangle \leq R$ for all feasible $X$), then we have \emph{strong duality}: the optimal values of $(P)$ and $(D)$ are equal and attained~\cite[Exercise~2.12]{BTN01}. If moreover $C,A_1,\ldots, A_m,b$ are rational and their binary encoding length  is at most $L$, then one can approximate either optimal value up to an additive error $\epsilon$ in time $\mathrm{poly}(N,m, \log(R/(r\epsilon)),L)$~\cite{GLS81,dKF16}.

The \emph{Gram matrix} of a set of vectors $x_1,\ldots, x_N\in \R^d$ is the $N \times N$ matrix whose entries are  
\[
\mathrm{Gram}(x_1,\ldots, x_N)_{i,j} := \langle x_i, x_j\rangle.
\]
 Recall that a matrix $X \in \mathcal S^N$ is positive semidefinite if and only if $X = \mathrm{Gram}(x_1, \ldots, x_N)$ for some vectors $x_1, \ldots, x_N\in \R^d$ (for some $d\in\N$). For two sets of vectors $\{x_1,\ldots, x_k\}$ and $\{y_1,\ldots, y_\ell\}$ we use the shorthand notation $\mathrm{Gram}(\{x_i\}, \{y_j\})$ for the Gram matrix of the $k+\ell$ vectors $x_1,\ldots, x_k, y_1,\ldots, y_\ell$, which has the  block structure: 
\[
\mathrm{Gram}(\{x_i\}, \{y_j\}) = \begin{pmatrix} \big(\langle x_i, x_j \rangle\big) \vspace{0.2cm}& \big(\langle x_i, y_{j} \rangle\big) \\ \big(\langle y_i, x_j\rangle\big) & \big(\langle y_i, y_j \rangle \big) \end{pmatrix}.
\]

We will use the following lemma, which follows from a well-known isometry property of the Euclidean space; we give a short proof for completeness.
\begin{lemma} \label{lem:elementaryfact}
Let $x_1,\ldots, x_k, y_1,\ldots, y_k \in \R^d$ for some $d \in \N$. If $\langle x_i, x_j \rangle = \langle y_i, y_j \rangle$ for all $i,j \in [k]$, then there exists a matrix $U \in O(d)$ such that $U x_i = y_i$ for all $i\in [k]$. 
\end{lemma}
\begin{proof}
We may assume that both sets $\{x_1,\ldots,x_k\}$ and $\{y_1,\ldots, y_k\}$ are linearly independent (since, for any $\lambda\in \R^k$, $\sum_i \lambda_ix_i=0$ if and only if $\sum_i\lambda_iy_i=0$, as $\|\sum_i \lambda_ix_i\|^2=\|\sum_i\lambda_i y_i\|^2$). We may also assume that $k=d$ (else consider vectors $x_{k+1},\ldots,x_d\in \R^d$ forming an orthonormal basis of $\text{Span}(x_1,\ldots,x_k)^\perp$ and analogously for the $y_i$'s). Now it follows from the assumption $(\langle x_i,x_j\rangle)_{i,j=1}^d=(\langle y_i,y_j\rangle)_{i,j=1}^d$ that the linear map $U$ such that $U x_i=y_i$ for $i\in [d]$ is orthogonal.
\end{proof}

Throughout we let $e$ denote the all-ones vector (of appropriate size) and set $[n]=\{1,\ldots,n\}$ for any integer $n\in\N$. For an integer $t$, we use the shorthand notation $\binom{n}{\leq t}$ for $\{ S \subseteq [n] : |S|\leq t\}$. In what follows we use tensors, matrices, and vectors indexed by tuples $(i_1,\ldots,i_t)\in [n]^t$. We will use the notation $\ind{i}{t}$ to denote such a tuple, whose length (here $t$) will be clear from the context, and we let $\ind{i}{} \, \ind{j}{}=(i_1,\ldots,i_t,j_1,\ldots,j_s)$  denote the concatenation of two tuples $\ind{i}{}=(i_1,\ldots,i_t)$ and $\ind{j}{}=(j_1,\ldots,j_s)$. We may view a tensor $T \in \R^{n \times \cdots \times n}$ either as a map from $[n] \times \cdots \times [n]$ to $\R$ given by $\ind{i}{t} \mapsto T_{\ind{i}{}}$, or as a multilinear form on $\R^n \times \cdots \times \R^n$ given by $(z_1,\ldots,z_t)\mapsto T(z_1,\ldots,z_t) = \sum_{i_1,\ldots, i_t=1}^{n} T_{i_1,\ldots, i_t} z_1(i_1) \cdots z_t(i_t)$. We use the $(n+1)$-dimensional \emph{Lorentz cone} $\mathcal L^{n+1}$, defined by $\mathcal L^{n+1} = \{ (w,v) \in \R \times \R^{n} : w \geq ||v||_2\}$.

\section{Semidefinite programs for the completely bounded norm}\label{secSDP}

In this section we provide  semidefinite programming reformulations of the completely bounded norm of a tensor. We first explain the main idea for building such a program, which essentially follows by using an adaptation of Lemma~\ref{lem:elementaryfact}, and then we indicate how to design a more economical SDP, using smaller matrices and less constraints.

\subsection{Basic construction for a semidefinite programming formulation}

Recall that the operator norm of a matrix $A$ is defined by $||A|| = \max_{v: ||v||=1} ||Av||$, or, equivalently, by $||A|| = \sup_{u,v: ||u|| = ||v|| = 1} \langle u, Av\rangle$. Using the latter definition we can reformulate the completely bounded norm $\|T\|_{\mathrm{cb}}$ of a $t$-tensor $T$ as the optimal value of the following program:
\begin{align}
\sup\Big\{\smash{\sum_{i_1,\ldots,i_{t} = 1}^n} T_{i_1,\ldots,i_{t}} \, \langle u, U_1(i_1)\cdots U_{t}(i_{t}) v \rangle : \  &d \in \N, \ u,v \in \R^d \text{ unit}, \  U_j(i) \in O(d) \text{ for } i \in [n], j \in [t] \Big\} \label{eqdef2}
\end{align}
We now show how to use Lemma~\ref{lem:elementaryfact} to characterize vectors of the form $U_1(i_1)\cdots U_{t}(i_{t}) v$, where $v$ is a unit vector and $U_j(i)$ are orthogonal matrices, in terms of their Gram matrix. 

\begin{lemma} \label{lem1}
Let $\{v_{\ind{i}{t}}\}_{\ind{i}{t} = (i_1,\ldots, i_t) \in [n]^t}$ be a set of unit vectors in $\R^d$. 
There exist orthogonal matrices $U_j(i) \in O(d)$ for $j \in [t], i \in [n]$ and a unit vector $v \in \R^d$ such that 
\[
v_{\ind{i}{t}} = U_1(i_1)\cdots U_t(i_t) v \qquad \text{for all } \ind{i}{t} = (i_1,\ldots,i_t) \in [n]^t,
\]
if and only if 
\begin{equation} \label{cond}
\langle v_{\ind{i}{\ell}\, \ind{j}{t-\ell}},  v_{\ind{i}{\ell}\, \ind{k}{t-\ell}}\rangle = \langle v_{\ind{i}{\ell}'\, \ind{j}{t-\ell}},  v_{\ind{i}{\ell}'\, \ind{k}{t-\ell}}\rangle \quad \text{for all } \ell \in [t-1], \text{ and indices }
\ind{i}{\ell}, \ind{i}{\ell}' \in [n]^\ell, \ind{j}{t-\ell}, \ind{k}{t-\ell} \in [n]^{t-\ell}. 
\end{equation}
\end{lemma}
\begin{proof}
The  `only if' part is easy: indeed  for any indices $\ind{i}{}\in [n]^\ell$ and 
$\ind{j}{}=(j_{\ell+1},\ldots,j_t),$ $ \ind{k}{}=(k_{\ell+1},\ldots,k_t)\in [n]^{t-\ell}$ we have
\[
\langle v_{\ind{i}{\ell} \, \ind{j}{t-\ell}},  v_{\ind{i}{\ell} \, \ind{k}{t-\ell}} \rangle = \langle U_{\ell+1}(j_{\ell+1})\cdots U_t(j_{t}) v, U_{\ell+1}(k_{\ell+1})\cdots U_t(k_{t}) v \rangle.
\]
We show the `if' part by induction on $t\geq1$. Assume first  $t=1$ (in which case  condition (\ref{cond}) is  void). 
By assumption, the vectors $v_1,\ldots, v_n$ are unit vectors.
Then pick a unit vector $v\in \R^d$ and for each $i\in [n]$ let $U(i)\in O(d)$ be such that $U(i) v = v_i$ (which exists by Lemma \ref{lem:elementaryfact}).  
Assume now that $t \geq 2$. 
Fix the index $1\in [n]$ and for any $i_1 \in [n]\setminus \{1\}$ consider the two  sets of vectors 
\[
\{v_{1 \, \ind{i}{t-1}}: \ind{i}{t-1} \in [n]^{t-1}\} \qquad \text{and} \qquad \{v_{i_1 \, \ind{i}{t-1}}: \ind{i}{t-1} \in [n]^{t-1}\}.
\] Observe that it follows from condition~\eqref{cond} (case $\ell=1$) that 
\[
\langle v_{1 \, \ind{j}{t-1}}, v_{1 \, \ind{k}{t-1}} \rangle = \langle v_{i_1 \, \ind{j}{t-1}}, v_{i_1 \, \ind{k}{t-1}} \rangle \qquad \text{ for all } \ind{j}{t-1},\ind{k}{t-1} \in [n]^{t-1}.
\]
Hence we may apply Lemma~\ref{lem:elementaryfact}: there exists an orthogonal matrix $U_1(i_1) \in O(d)$ such that 
$$v_{i_1 \, \ind{i}{t-1}} = U_1(i_1) v_{1 \, \ind{i}{t-1}} \quad \text{ for all }  \ind{i}{}\in [n]^{t-1}.$$
 We can now apply the induction hypothesis to the vectors $v_{1\,\ind{i}{t-1}}$ ($\ind{i}{t-1} \in [n]^{t-1}$). Since they satisfy~\eqref{cond} (with $t$ replaced by $t-1$) it follows that there exist orthogonal matrices $U_2(i_2),\ldots,U_t(i_t) \in O(d)$ and a unit vector $v \in \R^d$ such that 
$$ v_{1\,\ind{i}{t-1}} = U_2(i_2)\cdots U_t(i_t) v \quad \text{ for all }  \ind{i}{}\in[n]^{t-1}.$$
Combining the above two relations we obtain
\[
v_{i_1,\ldots,i_t} = U_1(i_1) v_{1\, \ind{i}{t-1}} = U_1(i_1) \cdots U_t(i_t) v \qquad \text{ for all } \ind{i}{t-1} \in [n]^{t-1}.  
\]
This concludes the proof. 
\end{proof}

We are now ready to give an equivalent SDP formulation for the program~\eqref{eqdef2}. 
First, in  view of Lemma \ref{lem1}, we can  rewrite~\eqref{eqdef2} as 
\begin{equation}\label{eqdef2b}
\sup\Big\{ \sum_{\ind{i}{}\in [n^t]} T_{\ind{i}{}} \, \langle u, v_{\ind{i}{}}\rangle: d\in \N, v_{\ind{i}{}}\in \R^d \text{ unit vectors satisfying } (\ref{cond})\Big\}.
\end{equation}
Consider now the Gram matrix of the vectors $u, v_{\ind{i}{}}$ (for $\ind{i}{}\in [n]^t$):
\[
X = \mathrm{Gram}(\{u\}, \{v_{\ind{i}{t}}\}_{\ind{i}{t} \in [n]^t}) \in \mathcal S^{1 + n^t}_+.
\]
Let $A_1,\ldots, A_{m_0} \in \mathcal S^{1 + n^t}$ be  matrices such that the linear constraints $\mathrm{Tr}(A_i X) = 0$ (for $i \in [m_0]$)  enforce  condition~\eqref{cond} on $X$ (namely, the fact that the entry $X_{\ind{i}{}\, \ind{j}{}, \ind{i}{}\,\ind{k}{}}$ does not depend on the choice of $\ind{i}{}$), and define the operator $\mathcal A_0(X)=(\Tr(A_1 X),\ldots,\Tr(A_{m_0} X))$. 
The number of these linear constraints is $m_0 = \sum_{\ell=1}^{t-1} (n^\ell-1) \binom{n^{t-\ell}}{2} \leq \binom{n^t}{2}$,  where the last inequality follows from the fact that each entry in the bottom-right $n^t \times n^t$ principal submatrix of $X$ appears in at most one equation. In addition, let $C_0(T) \in \mathcal S^{1+n^t}$ be the block-matrix whose first diagonal block is indexed by $0$ (corresponding to $u$) and whose second diagonal block is indexed by the tuples $\ind{i}{t} \in [n]^t$, defined as follows:
\begin{equation}\label{eqCT0}
C_0(T) = \frac{1}{2}\left(\begin{array}{c|c} 0 & \ldots \ T_{\ind{i}{t}}\  \ldots \\ \hline  \vdots & \\ T_{\ind{i}{t}} & { 0} \\ \vdots & \end{array}\right).
\end{equation}
It follows that
\[
\langle C_0(T),X\rangle =\sum_{\ind{i}{t} \in [n]^t} T_{\ind{i}{t}}\  X_{0, \ind{i}{t}}=\sum_{\ind{i}{t} \in [n]^t} T_{\ind{i}{t}} \ \langle u, v_{\ind{i}{}}\rangle
\]
is precisely the objective function in the program (\ref{eqdef2b}) (and thus of (\ref{eqdef2})).
Consider now the following pair of primal/dual semidefinite programs:
\begin{align} \label{eq:sdpcbt}
 \max \quad & \langle C_0(T), X \rangle &  \min \quad & \langle e, \lambda\rangle  \\ \notag
\text{s.t.} \quad& X \in \mathcal S_+^{1+n^t} &\text{s.t.} \quad&\lambda \in \R^{1+n^t}, y \in \R^{m_0} \\ \notag
& \mathrm{diag}(X) = e, \ \mathcal A_0(X) = 0 \qquad &&\mathrm{Diag}(\lambda) + \mathcal A_0^*(y) - C_0(T) \in \mathcal S^{1+n^t}_+ 
\end{align}
It follows from the above discussion that the optimal value of the primal problem equals $||T||_\mathrm{cb}$. Observe that both the primal and the dual are strictly feasible (for the primal the identity matrix provides a strictly feasible solution). Hence, strong duality holds and the optima in both primal and dual are equal and attained (justifying the use of max and min). In other words this shows:

\begin{theorem}
The completely bounded norm $\|T\|_{\mathrm{cb}}$ of a $t$-tensor $T$ acting on $\R^n$  is given by any of the two semidefinite programs in~\eqref{eq:sdpcbt}. Moreover, the supremum in definition~\eqref{eqdef} is attained and one may restrict the optimization to size $d \leq 1+n^t$. 
\end{theorem}

We also observe that the primal SDP in (\ref{eq:sdpcbt}) involves a matrix of size $O(n^t)$ and has $O(n^{2t})$ linear constraints. Hence,  for fixed $t$, the optimal values can be approximated up to an additive error $\epsilon$ in time polynomial in $n$ and $\log(1/\epsilon)$.

\subsection{Reducing the size of the semidefinite program}

To obtain  a more efficient semidefinite programming representation of the completely bounded norm of a $t$-tensor, we fix an integer $s\in [t]$ and use the following observation: 
\[
\langle u, U_1(i_1)\cdots U_{t}(i_{t}) v \rangle = \langle U_s(i_s)^* \cdots U_1(i_1)^* u, U_{s+1}(i_{s+1})\cdots U_{t}(i_{t}) v \rangle.
\]
We characterized in Lemma \ref{lem1} the vectors of the form $U_{s+1}(i_{s+1})\cdots U_{t}(i_{t}) v$, where $v\in \R^d$ is a unit vector and $U_j(i)\in O(d)$ for $i\in [n]$ and $j\in \{s+1,\ldots,t\}$, as the unit vectors $v_{\ind{b}{}}\in\R^d$ (for $\ind{b}{}\in [n]^{t-s}$) satisfying the condition:
\begin{equation}\label{condv}
\langle v_{\ind{i}{\ell}\, \ind{j}{t-\ell}},  v_{\ind{i}{\ell}\, \ind{k}{t-\ell}}\rangle = \langle v_{\ind{i}{\ell}'\, \ind{j}{t-\ell}},  v_{\ind{i}{\ell}'\, \ind{k}{t-\ell}}\rangle \quad \textit{for all } \ell \in [t-s-1], \textit{ and indices }
\ind{i}{\ell}, \ind{i}{\ell}' \in [n]^\ell, \ind{j}{t-\ell}, \ind{k}{t-\ell} \in [n]^{t-s-\ell}. 
\end{equation}
Analogously to Lemma \ref{lem1} we have the following characterization for the vectors of the form $U_s(i_s)^* \cdots U_1(i_1)^* u$, where 
$u\in \R^d$ is a a unit vector and $U_j(i)\in O(d)$ for $i\in [n]$ and $j\in [s]$. (Note that $U^*\in O(d)$ if and only if $U\in O(d)$).

\begin{lemma}
Let $\{u_{\ind{a}{t}}\}_{\ind{a}{t} = (i_1,\ldots, i_s) \in [n]^s}$ be a set of unit vectors in $\R^d$. There exist orthogonal matrices $U_j(i) \in O(d)$ for $j \in [s], i \in [n]$ and a unit vector $u \in \R^d$ such that 
\[
u_{\ind{a}{t}} = U_s(i_s)\cdots U_1(i_1) u \qquad \text{for all } \ind{a}{s} = (i_1,\ldots,i_s) \in [n]^s,
\]
if and only if 
\begin{equation} \label{cond2}
\langle u_{\ind{j}{\ell}\, \ind{i}{t-\ell}},  u_{\ind{k}{\ell}\, \ind{i}{t-\ell}}\rangle = \langle u_{\ind{j}{\ell}\, \ind{i}{t-\ell}'},  u_{\ind{k}{\ell}\, \ind{i}{t-\ell}'}\rangle \quad \text{for all } \ell \in [s-1], \text{ and indices }
\ind{i}{\ell}, \ind{i}{\ell}' \in [n]^\ell, \ind{j}{t-\ell}, \ind{k}{t-\ell} \in [n]^{s-\ell}.
\end{equation}
\end{lemma}

We can now rewrite the program~\eqref{eqdef2} as an SDP using matrices of size $n^s + n^{t-s}$. 
Indeed, by the above, program (\ref{eqdef2}) can be equivalently rewritten as 
\begin{equation}\label{eqdef2c}
\begin{array}{ll}
\sup\Big\{\smash{\sum_{\ind{a}{}\in [n]^s,\ind{b}{}\in [n]^{t-s}} T_{\ind{a}{}\,\ind{b}{}} \,\langle u_{\ind{a}{}}, v_{\ind{b}{}}\rangle }: \ \ &v_{\ind{b}{}}\in\R^d \text{ unit vectors satisfying}~\eqref{condv}, \\
&u_{\ind{a}{}} \in \R^d \text{ unit vectors satisfying}~\eqref{cond2}\ \ \Big\}.
\end{array}
\end{equation}
Consider now the Gram matrix of the vectors $\{u_{\ind{a}{}}\}$ and $\{v_{\ind{b}{}}\}$:
\[
X = \mathrm{Gram}(\{u_{\ind{a}{s}}\}_{\ind{a}{s} \in [n]^s}, \{v_{\ind{b}{t-s}}\}_{\ind{b}{t-s} \in [n]^{t-s}}),
\]
let $A_1,\ldots, A_{m_s} \in \mathcal S^{n^s + n^{t-s}}$ be matrices such that the linear constraints $\mathrm{Tr}(A_i X) = 0$ (for $i \in [m_s]$) enforce  the conditions~\eqref{condv} and~\eqref{cond2}, and define the operator $\mathcal A_s(X)=(\Tr(A_1 X),\ldots, \Tr(A_{m_s} X))$. 
Observe that $X$ has size $n^s+n^{t-s}$, which is minimized when selecting $s=\lfloor t/2\rfloor$.
Moreover, the number of linear constraints satisfies 
\begin{equation} \label{eq:mefficient}
m_s = \sum_{\ell=1}^{s-1} (n^\ell-1) \binom{n^{s-\ell}}{2} + \sum_{\ell=1}^{t-s-1} (n^\ell-1) \binom{n^{t-s-\ell}}{2} \leq \binom{n^s}{2}+\binom{n^{t-s}}{2} < \binom{n^t}{2}.
\end{equation}
Let $C_s(T)  \in \mathcal S^{n^s+n^{t-s}}$ be the block-matrix whose first diagonal block is indexed by tuples $\ind{a}{} \in [n]^{s}$ and whose second diagonal block is indexed by tuples $\ind{b}{t} \in [n]^{t-s}$, given by 
\begin{equation}\label{eqCTs}
C_s(T)= \frac{1}{2} \begin{pmatrix} 0 & M(T) \\ M(T)^* & 0 \end{pmatrix},
\end{equation}
where $M(T)\in\R^{n^s\times n^{t-s}}$ has entries $M(T)_{\ind{a}{t},\ind{b}{t}} := T_{\ind{a}{t} \, \ind{b}{}}$. 
Note that  when selecting $s=0$  the matrix in (\ref{eqCTs}) coincides with the matrix in (\ref{eqCT0}).
It follows that 
\[
\langle C_s(T),X\rangle =\sum_{\ind{a}{s} \in [n]^s,\, \ind{b}{t-s} \in [n]^{t-s}} T_{\ind{a}{s}\,\ind{b}{t-s}} X_{\ind{a}{s}, \ind{b}{t-s}}=
 \sum_{\ind{a}{s} \in [n]^s,\, \ind{b}{t-s} \in [n]^{t-s}} T_{\ind{a}{s}\,\ind{b}{t-s}} \,\langle u_{\ind{a}{}}, v_{\ind{b}{}}\rangle
\]
is  the objective function of the program (\ref{eqdef2c}) (and thus of (\ref{eqdef2})).
Then we can define the pair of primal/dual semidefinite programs
\begin{align} \label{eq:sdpcb}
\max \quad & \langle C_s(T), X \rangle & \min \quad & \langle e, \lambda\rangle  \\ \notag
\text{s.t.} \quad& X \in \mathcal S_+^{n^s+n^{t-s}} &\text{s.t.} \quad&\lambda \in \R^{n^s+n^{t-s}}, y \in \R^m \\ \notag
& \mathrm{diag}(X) = e, \ \mathcal A_s(X) = 0 \qquad &&\mathrm{Diag}(\lambda) + \mathcal A_s^*(y) - C_s(T) \in \mathcal S^{n^s+n^{t-s}}_+ 
\end{align}
whose optimal values provide as before the completely bounded norm  $\|T\|_{\mathrm{cb}}$.

\begin{theorem}\label{thm2}
The completely bounded norm $\|T\|_{\mathrm{cb}}$ of a $t$-tensor $T$ acting on $\R^n$  is given by any of the two semidefinite programs in~\eqref{eq:sdpcb}. Moreover, the supremum in definition~\eqref{eqdef} is attained and one may restrict the optimization to size $d \leq n^s + n^{t-s}$ for any integer $s\in [t]$.
\end{theorem}

If we select   $s=\lfloor t/2\rfloor$ the primal program in (\ref{eq:sdpcb}) involves a matrix variable of size $n^{\lfloor t/2 \rfloor}+n^{\lceil t/2\rceil}$ and it has $O(n^{2 \lceil t/2\rceil})$ affine constraints. This  represents a significant size reduction with respect to the program in (\ref{eq:sdpcbt}) (corresponding to the choice $s=0$), which involves a matrix variable of size $1+n^t$ and $O(n^{2t})$ affine constraints.

\section{SDP characterization of quantum query complexity of Boolean functions} \label{sec:qquerycompl}

In this section  we illustrate the relevance of the above results  through the connection established recently in~\cite{ABP17} between
the completely bounded norm of tensors and the quantum query complexity of Boolean functions. 
After a brief recap on the quantum query complexity of Boolean functions, we give a new SDP characterization for the quantum query complexity $Q_\epsilon(f)$ of a Boolean function $f$. We then compare our SDP to the known SDP characterizations. Finally we use our SDP to derive a new type of certificates for large quantum query complexity: $Q_\epsilon(f)>t$. 

\subsection{Quantum query complexity }

We are given a domain $D \subseteq \cube$ and a Boolean function $f: D \to \{\pm 1\}$. The function is called {\em total} when $D=\{\pm 1\}^n$ and {\em partial} otherwise. 
The  task is to compute the value $f(x)$ for an input $x\in D$ while having access to~$x$ only through some oracle. The objective is to compute $f(x)$ using the smallest possible number of oracle calls on a worst-case input $x$, and the least such number is called the \emph{classical/quantum query complexity} of the function $f$. See for instance~\cite{BW02} for a survey on query complexity and~\cite{ABK16} for a more recent overview of the relation between classical and quantum query complexity.  

In the classical case, the oracle consists of querying the value of the entry $x_i$ for a selected index $i\in [n]$. 
In the quantum case, an oracle query to $x$ is defined as an application of the \emph{phase oracle} $O_x$, which is the diagonal unitary operator acting on $\C^{n+1}$  defined by $O_x=\text{Diag}(x_1,\ldots,x_n,1)$.\footnote{The above described classical oracle can be seen as applying the function $C_x:[n]\times \{\pm 1\} \rightarrow [n] \times \{\pm 1\}$, defined by $(i,b) \mapsto (i,x_i b)$, to the input $(i,1)$, and it is called the \emph{standard bit oracle}. In the quantum setting $C_x$ corresponds to applying the operator $\Diag(x) \oplus - \Diag(x)$. For this section, in the quantum setting it will be more convenient to work with the phase oracle $O_x = \Diag(x,1)$. It is well known that for quantum query algorithms the two oracles are equivalent, see also~\cite{BSS03,AAIKS16}.} 
A {\em $t$-query quantum algorithm} can be described by a Hilbert space $\Hilbert=\C^{n+1}\otimes \C^d$ (for some $d\in\N$), 
 a sequence of unitaries $U_0,\ldots, U_t$ acting on $\Hilbert$, two Hermitian positive semidefinite operators $P_{+1},P_{-1}$ on $\Hilbert$ satisfying $P_{+1}+P_{-1} = I$, and a unit vector $v\in \Hilbert$. The algorithm starts in the state $v$ and alternates between applying a unitary $U_j$ and the oracle $O_x$. The final state of the algorithm on input $x \in D$ is  
\[
\psi_x := U_t (O_x \otimes I_d) U_{t-1} (O_x \otimes I_d) U_{t-2} \cdots U_1 (O_x \otimes I_d)U_0 v.
\]
The algorithm concludes by `measuring' $\psi_x$ with respect to  $\{P_{+1},P_{-1}\}$, which means that it outputs $+1$ with probability $\psi_x^* P_{+1} \psi_x$ and $-1$ with probability $\psi_x^* P_{-1} \psi_x$. 
The expected output of the algorithm is therefore given by 
\begin{equation} \label{eq:outcome}
    \psi_x^* (P_{+1} - P_{-1})  \psi_x.
\end{equation}
Given $\epsilon \geq 0$ the {\em bounded-error quantum query complexity} of $f$, denoted as $Q_\epsilon(f)$, is the smallest number of queries a quantum algorithm must make on the worst-case input $x\in  D$ to compute $f(x)$ with probability at least $1-\epsilon$. Note that for a quantum algorithm which computes $f(x)$ with probability at least $1-\epsilon$ we have $|\psi_x^* (P_{+1} - P_{-1})  \psi_x - f(x)| \leq 2\epsilon$.

Determining the quantum query complexity of a given function $f$ is non-trivial. We mention two lines of work that have been developed in the recent years:
the \emph{polynomial method} from~\cite{BBCMdW98} and  the \emph{adversary method} from~\cite{Amb02}. Both methods have been used to provide lower bounds on the quantum query complexity. 
The adversary method was strengthened in~\cite{HLS07} to the \emph{general adversary method}, which provides a parameter $\mathrm{ADV}^\pm(f)$, satisfying $Q_{\epsilon}(f) = \Theta(\mathrm{ADV}^\pm(f))$ for any fixed $\epsilon \in (0,\frac{1}{2})$~\cite{HLS07,Rei09,Rei11,LMRSS11}.\footnote{To be more precise, for all $\epsilon \in (0,\frac{1}{2})$, we have $\frac{1- 2 \sqrt{\epsilon(1-\epsilon)}}{2} \mathrm{ADV}^\pm(f) \leq Q_\epsilon(f) \leq 100 \frac{\log(1/\epsilon)}{\epsilon^2} \mathrm{ADV}^\pm(f)$, where the first inequality is shown in~\cite{HLS07} and the second one in~\cite{LMRSS11}.} 
The polynomial method has been strengthened only very recently in \cite{ABP17} and is shown there to provide an exact characterization of the quantum query complexity. Since the polynomial method is the most relevant  to our work we explain it in some more detail below. 

\medskip
The polynomial method is based on  the observation made in~\cite{BBCMdW98} that~\eqref{eq:outcome} in fact defines an 
$n$-variate polynomial $p$ with degree at most $2t$ such that $p(x) = \psi_x^* (P_{+1} - P_{-1}) \psi_x$ equals the expected value of the returned sign of the algorithm for an input $x\in D$. This motivated considering the {\em $\epsilon$-approximate degree of $f$}, $\deg_\epsilon(f)$, defined by
\begin{align} \label{eqdeg}
\deg_\epsilon(f) = \min \quad &t \\ \notag
\text{s.t.}\quad & \exists \, n\text{-variate polynomial } p \text{ with } \deg(p)\le t \\ \notag
 & |p(x)-f(x)|\le 2\epsilon \quad\forall x\in D,\\ \notag
 & |p(x)|\le 1+2\epsilon \quad \forall x\in\{\pm 1\}^n.
\end{align}
Then, as shown in~\cite{BBCMdW98}, it follows from the above that the scaled approximate degree of $f$ is a lower bound on $Q_\epsilon(f)$: 
\[
\deg_\epsilon(f)\le 2Q_\epsilon(f).
\]

In~\cite{AA15} (see also~\cite{AAIKS16}) the observation is made that~\eqref{eq:outcome} can be used to define a $2t$-tensor $T \in \R^{(n+1) \times \cdots (n+1)}$ by using different input strings at the successive queries. More precisely, for any $(z_1,\ldots,z_{2t})\in \R^{n+1}\times \ldots \times \R^{n+1}$, we can define
\begin{equation} \label{eq:defT}
    T(z_1,\ldots, z_{2t}) = v^* U_0^* \widetilde O_{z_1} \cdots \widetilde O_{z_t} U_t^* (P_{+1} - P_{-1})  U_t \widetilde O_{z_{t+1}} U_{t-1} \widetilde O_{z_{t+2}} \cdots \widetilde O_{z_{2t}} U_0 v,
\end{equation}
where $\widetilde O_z = O_z \otimes I_d$, so that $T((x,1),\ldots,(x,1))$ equals the expected value of the returned sign of the quantum algorithm for all $x\in D$. Note that $T$ is in fact bounded on the entire hypercube: $T$ satisfies the inequalities $|T(z_1,\ldots, z_{2t})| \leq 1$ for all $z_1,\ldots, z_{2t} \in \{\pm 1\}^{n+1}$. This led to the following notion of  \emph{block-multilinear approximate degree}, $\bmdeg(f)$, defined by
\begin{align} \label{eqbmdeg}
\bmdeg(f) = \min \quad &t \\ \notag
\text{s.t.}\quad & \exists \, t\text{-tensor } T \text{ acting on } \R^{n+1}, \\ \notag
& | T((x,1),\ldots,(x,1))-f(x)| \leq 2\epsilon \qquad \forall x\in D,\\ \notag
& |T(z_1,\ldots,z_t)|\le 1\quad\forall z_1,\ldots,z_t\in \{\pm 1\}^{n+1}.
\end{align}
Note that if $T$ is a $t$-tensor that is feasible for the program~\eqref{eqbmdeg}, then the polynomial $p$ defined by $p(x)=T((x,1),\ldots,(x,1))$ is feasible for the program~\eqref{eqdeg}, and thus we have
$$\deg_\epsilon(f)\leq \bmdeg(f)\leq 2Q_\epsilon(f).$$

In the recent work~\cite{ABP17} it is shown that the $2t$-tensor in~\eqref{eq:defT} in fact has completely bounded norm at most $1$. In addition, the authors of \cite{ABP17}  also show the converse: the existence of a $2t$-tensor $T \in \R^{(n+1) \times \cdots \times (n+1)}$ that satisfies $||T||\cb\leq 1$ and  $|T((x,1),\ldots,(x,1))-f(x)| \leq 2\epsilon$ for all $x \in D$, ensures the existence of a $t$-query quantum algorithm that outputs the correct sign with probability at least $1-\epsilon$ for all $x \in D$.\footnote{In~\cite{ABP17} the result is stated using a tensor $T \in \R^{2n \times \cdots \times 2n}$. The dimension $2n$ corresponds to the fact that a controlled-phase gate (acting on $\C^{2n}$) is unitarily equivalent to the standard bit oracle. When we allow the quantum algorithm to use additional workspace we obtain the same query complexity measure working with the oracle $O_x = \Diag(x,1)$ (acting on $\C^{n+1}$). See also~\cite{BSS03,AAIKS16}.} That is, if such a $2t$-tensor exists, then $Q_\epsilon(f) \leq t$. 
This leads to the notion of \emph{completely bounded approximate degree}, $\cbdeg(f)$, defined by\footnote{Note that~\cite{ABP17} use the same definition except that they consider the least $t$ such that there exists a $2t$-tensor with these properties.} 
\begin{align} \label{eqcbdeg}
\cbdeg(f) = \min \quad &t \\ \notag
\text{s.t.}\quad & \exists \, t\text{-tensor } T \text{ acting on } \R^{n+1}, \\ \notag
& | T((x,1),\ldots,(x,1))-f(x)| \leq 2\epsilon \qquad \forall x\in D,\\ \notag
& ||T||\cb \leq 1.
\end{align}
Notice that $||T||\cb \leq 1$ implies that $|T(z_1,\ldots,z_{t})| \leq 1$ for all $z_1,\ldots z_t \in \{\pm 1\}^{n+1}$. Therefore we have
$$\deg_\epsilon(f)\leq \bmdeg(f)\leq \cbdeg(f)\leq 2Q_\epsilon(f).$$
As mentioned above, the last inequality is in fact an equality up to rounding: 
\begin{theorem}[{\cite[Cor.~1.5]{ABP17}}] \label{cor1}
For  a Boolean function $f: D \rightarrow \{-1,1\}$ and $\epsilon \geq 0$, we have 
\[
Q_\epsilon(f) = \lceil\cbdeg(f)/2\rceil. 
\]
\end{theorem}

\subsection{New semidefinite reformulation}

Using our earlier results in Section~\ref{secSDP} about the completely bounded norm of a tensor, we can express the completely bounded approximate degree $\cbdeg(f)$ using semidefinite programming. To certify the inequality $||T||\cb \leq 1$ we can use the dual SDP in~\eqref{eq:sdpcb} as follows:
\[
\|T\|_{\mathrm{cb}}\le 1 \Longleftrightarrow \exists \lambda \in \R^{N_s}, y\in \R^{m_s} \text{ such that } \langle e,\lambda \rangle \le 1 \text{ and } 
\Diag(\lambda)+\mathcal A_{s}^*(y)-C_s(T) \in \mathcal S_+^{N_s}.
\]
Here, we may choose $s$ to be any integer $0 \leq s \leq \lfloor t/2\rfloor$, so that $N_s$ is given by $(n+1)^{s} + (n+1)^{t-s}$ and $m_s$ by~\eqref{eq:mefficient} (with $n$ replaced by $n+1$).
We may then use the fact that the constraints: $|T((x,1),\ldots,(x,1))-f(x)| \leq 2\epsilon$ for all $x \in D$, can be written as linear constraints on the coefficients of $T$ to reformulate~\eqref{eqcbdeg} using semidefinite programming. To make it more apparent that $T((x,1),\ldots,(x,1))$ is a linear combination of the coefficients of $T$, recall that by definition $T(z,\ldots,z) = \sum_{i_1,\ldots,i_{t}=1}^{n+1} T_{i_1,\ldots,i_{t}} z_{i_1}\cdots z_{i_{t}}$ for all $z \in \R^{n+1}$. It follows that the parameter $\cbdeg(f)$ can be reformulated as the smallest integer $t\in \N$ for which the following SDP admits a feasible solution:
\begin{align} \label{QQuerySDP}
\cbdeg(f) = \min \quad &t \\ \notag
\text{s.t.}\quad & \exists \, t\text{-tensor } T \in \R^{(n+1)\times \ldots \times (n+1)},  \lambda \in \R^{(n+1)^{s} + (n+1)^{t-s}}, y \in \R^{m_{s}} \\ \notag
 &  \Big|\sum_{i_1,\ldots,i_{t}=1}^{n+1} T_{i_1,\ldots,i_{t}} z_{i_1}\cdots z_{i_{t}} - f(x)\Big| \le 2\epsilon   \qquad \text{ for } x \in D,\, z = (x,1) \in \R^{n+1}
 \\ \notag 
 & \langle e, \lambda \rangle \le 1   \\ \notag
& \diag(\lambda) + \mathcal A_s^* (y) - C_s(T) \in \mathcal S_+^{(n+1)^{s} + (n+1)^{t-s}}.
\end{align}
Recall that, due to Theorem~\ref{thm2}, we may choose $s$ to be any integer $0 \leq s \leq \lfloor t/2 \rfloor$. 

\subsection{Relation to the known SDPs for quantum query complexity} \label{sec:comparison} 

The above SDP~\eqref{QQuerySDP} is not the first SDP that characterizes the quantum query complexity. The parameter $\mathrm{ADV}^\pm(f)$ provided by the general adversary method in~\cite{HLS07} mentioned in the previous section can also be written as an SDP.
 Barnum, Saks, and Szegedy~\cite{BSS03} formulated another SDP characterization for the quantum query complexity. 
 Like ours, the SDP in \cite{BSS03}  expresses  $Q_\epsilon(f)$ as the smallest integer $t \in [n]$ for which there exist some positive semidefinite matrices satisfying some linear (in)equalities. 
 Both the Barnum-Saks-Szegedy SDP and the general adversary method SDP are derived by considering the behaviour of a quantum algorithm on pairs of different inputs; the matrix variables should be seen as the Gram matrices of vectors associated to the quantum algorithm. Instead, as explained before, our SDP fits in the framework of the polynomial method where we only consider the expected output of the quantum algorithm on different inputs.   
There are three main differences between these three SDP characterizations that we will highlight below. 
 
 First, solutions to either the Barnum-Saks-Szegedy SDP or the general adversary method SDP can be turned into quantum query algorithms, while a solution to our SDP only proves the existence of a quantum algorithm (it is not clear how to directly derive a quantum algorithm from it, as far as we know). 

A second difference is the size of the matrix variables involved in the various SDPs, which we have summarized in Table~\ref{tab:1}. We want to highlight the difference in block size between the three SDPs. Using our SDP one can certify the quantum query complexity $t$ of a Boolean function using a single matrix of size $\Theta(n^t)$, while both the general adversary method SDP and the Barnum-Saks-Szegedy SDP use several matrix variables of size $|D|$ (which is $2^n$ for total functions).
We mention that for $\epsilon=0$ our SDP for $\cbdeg(f)$ simplifies: the matrix variable remains of size $\Theta(n^{t})$, there is only one linear inequality, and the number of linear equalities remains $\Theta(n^{2t})$. This is because, since the equations $\sum_{i_1,\ldots,i_{t}=1}^{n+1} T_{i_1,\ldots,i_{t}} z_{i_1}\cdots z_{i_{t}} = f(x)$ (for $x \in D, z = (x,1)$) involve $\Theta(n^t)$ real variables (the coefficients of $T$), there are at most $\Theta(n^t)$ of linear equalities that are linearly independent, and we only need to impose linearly independent equality constraints. 

\begin{table}[ht]
    \centering
    \begin{tabular}{c|c|c|c|c}
         & \# blocks & block size & \# lin.~ineq. & \# equations  \\ \hline
         General adversary method & $n$ & $|D|$ & 0 & $|f^{-1}(1)| \cdot |f^{-1}(0)|$ \\
         Barnum-Saks-Szegedy & $nt+2$ & $|D|$ & $|D|$ & $\Theta(t \cdot |D|^2)$ \\
         Thrm~\ref{cor1} + $\cbdeg(f)$ & $1$ & $\Theta(n^t)$ & $2|D|+1$ & $\Theta(n^{2t})$
    \end{tabular}
    \captionsetup{width=.9\linewidth}
    \caption{A comparison of the size of the general adversary method SDP, the Barnum-Saks-Szegedy SDP, and our SDP for $\cbdeg(f)$. The latter two are feasibility problems whose size depends on the number of queries $t$ (which means we consider $\cbdeg(f) = 2t$). When viewed as a block-diagonal SDP, the first column specifies the number of blocks and the second one the size of the blocks, the third column gives the number of linear inequalities on entries of these blocks and the fourth one the number of linear equations.} 
    \label{tab:1}
\end{table}

A third difference is the fact that the adversary method SDP characterizes the quantum query complexity only up to a multiplicative factor, while both our SDP and the Barnum-Saks-Szegedy SDP give an exact characterization.

\subsection{Dual SDP certificates for large quantum query complexity}

We now turn our attention to providing lower bounds on the quantum query complexity. Given a fixed integer $t\in \N$, finding the smallest scalar $\epsilon \geq 0$ such that $\deg_\epsilon(f) \leq t$ can be expressed as a linear program. The duality theory of linear programming can therefore be used to provide tight lower bounds on the approximate degree $\deg_\epsilon(f)$. Likewise, as we explain in this section, our SDP formulation of $\cbdeg(f)$ and the duality theory of semidefinite programming can be used to give tight lower bounds on $\cbdeg(f)$. 

This section is organized as follows. We first rewrite the program expressing $\cbdeg(f)$ in a form that permits to give certificates for $\cbdeg(f)>t$. These certificates take the form of feasible solutions to a certain SDP. We show how linear programming certificates for $\deg_\epsilon(f)$ can be seen as SDP solutions with a specific structure. We then define an intermediate class of certificates based on second-order cone programming.  

\subsubsection{SDP certificates for $\cbdeg(f)>t$} \label{sdpreformulation}

Let $D \subseteq \cube$ and let $f: D \to \{\pm 1\}$. A certificate for $\cbdeg(f)>t$ can be given as follows. Consider the following minimization problem (derived from the program~\eqref{QQuerySDP}, setting $s=0$)\footnote{Note that we use the less efficient, but easier, SDP-formulation of the completely bounded norm~\eqref{eq:sdpcbt}. In~\eqref{QQuerySDP2} we could have just as easily used the more efficient formulation, but in Section~\ref{sec:certificates} it will be convenient to work with the less efficient formulation.}
\begin{align} \label{QQuerySDP2}
 \min \quad &2\epsilon \\ \notag
\text{s.t.}\quad & T \in \R^{(n+1)\times \ldots \times (n+1)} \text{ a $t$-tensor},  \lambda \in \R^{1+(n+1)^{t}}, y \in \R^m, \epsilon \in \R \\ \notag
 &  \Big|\sum_{i_1,\ldots,i_{t}=1}^{n+1} T_{i_1,\ldots,i_{t}} z_{i_1}\cdots z_{i_{t}} - f(x)\Big| \le 2\epsilon   \quad \text{for all } x \in D, z = (x,1) \in \R^{n+1}
 \\ \notag 
 & \langle e, \lambda \rangle \le 1   \\ \notag
& \diag(\lambda) + \mathcal A_0^* (y) - C_0(T) \in \mathcal S_+^{1+(n+1)^{t}}
\end{align}
Using semidefinite programming duality theory we can formulate its dual. After simplification the dual reads as follows:
\begin{align} \label{QQErrSDPDUAL2}
\max \quad &- w + \sum_{x \in D} \phi(x) f(x)  \\ \notag
\text{s.t.}\quad & \phi=(\phi(x))_{x\in D} \in \R^D,   X \in \mathcal S^{1+(n+1)^{t}}_+, w \in \R \\ \notag
&\sum_{x \in D} |\phi(x)| = 1 \\ \notag
&\mathrm{diag}(X) = w \cdot e   \\ \notag
&\mathcal A_0(X) = \mathbf 0  \\ \notag
&X_{0, \underline{i}} = \sum_{\substack{x \in D \\ z=(x,1)}} \phi(x) z_{i_1}\cdots z_{i_{t}} \quad \text{ for all } \underline{i}=(i_1,\ldots, i_{t}) \in [n+1]^{t}
\end{align}
Note that there is no duality gap since the dual program (\ref{QQErrSDPDUAL2}) is strictly feasible.
A  tuple $(\phi, X,w)$ that forms a feasible solution to~\eqref{QQErrSDPDUAL2} with objective value strictly larger than $2\epsilon$ is an {\emph{SDP certificate}} for $\cbdeg(f) > t$.  

We remark that, for total functions, i.e., when $D=\{\pm 1\}^n$,  the constraint on $X_{0,\underline{i}}$ says that $X_{0,\underline{i}}$ should be equal to a certain Fourier coefficient of the function $\phi$. To see this we briefly recall the basic relevant facts  of Fourier analysis on the Boolean cube. 

\paragraph{Fourier analysis on the Boolean cube.}
For functions $f,g:\{\pm 1\}^n \to \R$ we define the inner product $\langle f,g \rangle = \frac{1}{2^n} \sum_{x \in \{\pm1\}^n} f(x)g(x)$. Then the  character functions $\chi_S(x) := \prod_{i \in S} x_i$ ($S \subseteq [n]$) form an orthonormal basis with respect to this inner product. Any function $f:\{\pm 1\}^n\to\R$ can be expressed in this basis as $f(x) = \sum_{S \subseteq [n]} \widehat f(S) \chi_S(x)$, where $\widehat f(S) = \langle f, \chi_S\rangle$ are the Fourier coefficients of $f$. 

\medskip
We need one more definition in order to point out a link between the last constraint of program~\eqref{QQErrSDPDUAL2} and the Fourier coefficients of $\phi$.
For a tuple $\underline i=(i_1,\ldots,i_t)\in [n+1]^t$ let $S_{\underline i}$ denote the set of indices $k\in [n]$ that occur an odd number of times within the multiset $\{i_1,\ldots,i_t\}$.
Note the identity
\begin{equation} \label{eqiD}
\sum_{x\in D, z=(x,1)} \phi(x)z_{i_1}\cdots z_{i_t} =\sum _{x\in D} \phi(x) \prod_{k\in S_{\underline i}}x_k= \sum_{x\in D} \phi(x)\chi_{S_{\underline i}}(x).
\end{equation}
Hence, in the case $D=\{\pm 1\}^n$, we have
\begin{equation} \label{equi}
\sum_{x\in \{\pm 1\}^n, z=(x,1)} \phi(x)z_{i_1}\cdots z_{i_t} = 2^n\widehat \phi(S_{\underline i})
\end{equation}
and thus the last constraint in program~\eqref{QQErrSDPDUAL2} says that $X_{0,\underline i}= 2^n \widehat\phi(S_{\underline i})$ for all $\underline i\in [n+1]^t$.

\paragraph{Adding redundant inequalities to~\eqref{QQuerySDP2}.}

In the next section we will show how the SDP certificates for the completely bounded approximate degree generalize the linear programming certificates corresponding to the approximate degree. To do so, it will be useful to state an equivalent form of~\eqref{QQErrSDPDUAL2} derived by adding redundant inequalities to the primal problem. Recall that for a tensor $T$ the constraint $||T||\cb \leq 1$ implies that $\Big|\sum_{i_1,\ldots,i_{t}=1}^{n+1} T_{i_1,\ldots,i_{t}} z_{i_1}\cdots z_{i_{t}} \Big| = |T(z,\ldots,z)| \leq 1$ for all $z \in \{\pm 1\}^{n+1}$. As the last two constraints of~\eqref{QQuerySDP2} ensure $||T||\cb \leq 1$, it follows that the conditions 
\begin{equation} \label{eq:redundantineq}
\Big|\sum_{i_1,\ldots,i_{t}=1}^{n+1} T_{i_1,\ldots,i_{t}} z_{i_1}\cdots z_{i_{t}} \Big| \leq 1+ 2 \epsilon \text{ for all } x \not \in D, z = (x,1)
\end{equation}
are redundant for~\eqref{QQuerySDP2}. If we add these inequalities to~\eqref{QQuerySDP2} and then take the dual, then we obtain
\begin{align} \label{QQErrSDPDUAL2b}
\max \quad &- w +\sum_{x \in D} \phi(x) f(x) - \sum_{x \not \in D} |\phi(x)|  \\ \notag
\text{s.t.}\quad & \phi=(\phi(x))_{x\in \cube} \in \R^{\cube},  X \in \mathcal S^{1+(n+1)^{t}}_+, w \in \R \\ \notag
&\sum_{x \in \cube} |\phi(x)| = 1 \\ \notag
&\mathrm{diag}(X) = w \cdot e   \\ \notag
&\mathcal A_0(X) = \mathbf 0  \\ \notag
&X_{0, \underline{i}} = 2^n \widehat \phi(S_{\underline i}) \quad \text{ for all } \underline{i}=(i_1,\ldots, i_{t}) \in [n+1]^{t}
\end{align}
Notice that strong duality holds between the above program~\eqref{QQErrSDPDUAL2b} and the program defined by~\eqref{QQuerySDP2} and~\eqref{eq:redundantineq}. In particular it follows that the optimal value of program~\eqref{QQErrSDPDUAL2b} equals that of program~\eqref{QQErrSDPDUAL2}. Using complementary slackness, we can say slightly more when the optimal value is strictly positive.
\begin{lemma}
If the optimal value of the program~\eqref{QQErrSDPDUAL2b} is strictly positive, then any optimal solution $(\phi,X,w)$ to~\eqref{QQErrSDPDUAL2b} satisfies $\phi(x)=0$ for $x \not \in D$. 
\end{lemma}
\begin{proof}
Suppose that the above program~\eqref{QQErrSDPDUAL2b} has an optimal solution $(\phi,X,w)$ with objective value strictly positive. Then, by strong duality, the program defined by~\eqref{QQuerySDP2} and~\eqref{eq:redundantineq} has an optimal solution $(T,\lambda,y,\epsilon)$ with $\epsilon>0$. Since $\epsilon$ is strictly positive, there will be a strictly positive slack in the  inequalities~\eqref{eq:redundantineq}. By complementary slackness this means that the variables in the dual corresponding to these inequalities must be equal to zero for an optimal solution, that is, $\phi(x)=0$ for all $x \not \in D$. 
\end{proof}
Note that any tuple $(\phi,X,w)$ that is feasible for the program in~\eqref{QQErrSDPDUAL2b} and satisfies $\phi(x)=0$ for $x \not \in D$ is in fact feasible for the program in~\eqref{QQErrSDPDUAL2}.

\subsubsection{Approximate degree} \label{sec:approxdeg}
Let $f: D \to \{\pm 1\}$ be given. Given a fixed degree $t$, the smallest $\epsilon \geq 0$ for which there exists a polynomial $p$ of degree at most $t$ that satisfies $\sup_{x \in D} |f(x)-p(x)| \leq 2\epsilon$ and $\sup_{x \not \in D} |p(x)| \leq 1+ 2\epsilon$ can be determined using the following linear program:  
\begin{align} \label{eq:approxdeglp}
 \min \quad & 2\epsilon &  \max \quad & \sum_{x \in D} f(x) \phi(x) -\sum_{x \not \in D} |\phi(x)|  \\ \notag
\text{s.t.} \quad&  \Big|\sum_{S \in \binom{n}{\leq t}} c_S \chi_S(x) - f(x)\Big| \le 2\epsilon \quad \text{ for } x \in D  &\text{s.t.} \quad&\sum_{x \in \cube} |\phi(x)| = 1 \\ \notag
& \Big|\sum_{S \in \binom{n}{\leq t}} c_S \chi_S(x) \Big| \le 1+2\epsilon \quad \text{ for } x \not \in D  &&\phi(x) \in \R \quad \text{for } x \in \cube  \\ \notag
&c = (c_S) \in \smash{\R^{\binom{n}{\leq t}}}, \epsilon \in \R  &&\widehat \phi(S) = 0 \quad \text{for }  S \in \smash{\binom{n}{\leq t}}.
\end{align}
For a fixed $\epsilon \geq 0$, a polynomial $\phi$ that is a feasible solution to the maximization problem in~\eqref{eq:approxdeglp} with objective value strictly larger than $2\epsilon$ is called a \emph{dual polynomial} for $f$, and it is a certificate for $\deg_\epsilon(f) > t$. Note that by LP duality such a certificate exists whenever $\deg_\epsilon(f)>t$. Dual polynomials have been used to give tight bounds on the approximate degree of many Boolean functions, see for example~\cite{Spa08,She13,BT13,BKT17}. 

\bigskip

Feasible solutions to the maximization problem in~\eqref{eq:approxdeglp} provide feasible solutions to the SDP in~\eqref{QQErrSDPDUAL2b}. This gives a ``direct'' proof that dual polynomials give lower bounds on quantum query complexity. 
\begin{lemma} \label{lemapproxdeg}
Let $f:D \rightarrow \R$ and let $\phi: \cube \to \R$ be a feasible solution to~\eqref{eq:approxdeglp} with objective value strictly larger than $2 \epsilon$, then $\cbdeg(f) > t$.
\end{lemma}
\begin{proof}
Observe that the tuple $(\phi,X=0,w=0)$ forms a feasible solution to~\eqref{QQErrSDPDUAL2} with objective value strictly larger than $\epsilon$. Indeed, $X=0$ is positive semidefinite, it satisfies $\mathrm{diag}(X) = 0 = w \cdot e$ and $\mathcal A_0(X) = 0$. Moreover, the condition $\widehat \phi(S) = 0$ for all $S \in \binom{n}{\leq t}$ ensures that, for all $\underline i \in [n+1]^t$,
\[
X_{0, \underline{i}} = 2^n\widehat \phi(S_{\underline i})  = 0
\] 
since $|S_{\underline i}|\leq t$. 
\end{proof}

\subsubsection{New certificates for $\cbdeg(f) >t$ through second-order cone programming} \label{sec:certificates}
In (the proof of) Lemma~\ref{lemapproxdeg} we have seen that the linear programming certificates of $\deg_\epsilon(f)>t$ correspond to SDP certificates $(\phi,X,w) = (\phi,0,0)$ using the all-zeroes matrix $X=0$ in~\eqref{QQErrSDPDUAL2b}. Here we consider a more general class of SDP certificates $(\phi,X,w)$ where $X$ and $w$ still have an easy structure: those certificates for which we can take $X = \left(\begin{smallmatrix} w & v^T \\ v & w I \end{smallmatrix}\right)$ for some vector $v \in \R^{(n+1)^{t}}$ and real number $w$. This is based on the following observation.
\begin{lemma} \label{socpfeas}
Let $w \in \R$ and $v \in \R^{(n+1)^t}$. The matrix $X = \begin{pmatrix} w & v^T \\ v & w I \end{pmatrix}$ satisfies $\mathcal A(X) = 0$. Moreover, $X \in \mathcal S_+^{1+(n+1)^{t}}$ if and only if $w \geq ||v||_2$. 
\end{lemma}
\begin{proof}
First note that $\mathcal A_0(X) = 0$ is trivially satisfied by $X$. Indeed, $\mathcal A$ ignores the first row and column of $X$ and, for all $\ell \in [t-1]$, $\underline i, \underline i' \in [n+1]^\ell, \underline j, \underline k \in [n+1]^{t-\ell}$, we have that 
\[
X_{\underline i \, \underline j, \underline i \, \underline k}  = X_{\underline i' \, \underline j, \underline i' \, \underline k} = \begin{cases} w & \text{if } \underline j = \underline k \\ 0 & \text{else}. \end{cases}\]

Second, by considering the Schur complement of $X$ with respect to its upper-left corner, we have $X \in \mathcal S_+^{1+(n+1)^{t}}$ if and only if either $X = 0$, or $w>0$ and $w - v^T v/w \geq 0$.
\end{proof}

By restricting our attention to feasible solutions of the above form, the program~\eqref{QQErrSDPDUAL2b} reduces to the following second-order cone program:
\begin{align} \label{socp}
    \max \quad &-w+\sum_{x \in D} \phi(x) f(x) - \sum_{x \notin D} |\phi(x)|  \\ \notag
    \mathrm{s.t.} \quad &\sum_{x \in \cube} |\phi(x)| = 1,\ \ w \geq 2^n \sqrt{\sum_{\underline i \in [n+1]^{t}} \widehat \phi(S_{\underline i})^2}
\end{align}
This second-order cone program involves the $(1+(n+1)^t)$-dimensional Lorentz cone. However, by counting the number of tuples $\underline i$ for which $S_{\underline i}$ equals a given set $S$ we can reduce to dimension $\left| \binom{n}{\leq t} \right| = \sum_{k=0}^t \binom{n}{k}$. 
Indeed, for each subset $S \subseteq [n]$ with $|S| \leq t$ let $\mathcal I_S$ denote the set of tuples $\underline i \in [n+1]^t$ for which $S_{\underline i} = S$. One can verify that 
\[
|\mathcal I_S| = \sum_{\substack{k_1,\ldots, k_n \in \N, \\ \sum_{i} k_i \leq t, \\ k_l \text{ is odd for } l \in S, \\ k_l \text{ is even for } l \not \in S}} \frac{t!}{k_1! \cdots k_n! (t-\sum_i k_i)!}.
\]
Then by construction
\[
\sum_{\underline i \in [n+1]^{t}} \widehat \phi(S_{\underline i})^2 = \sum_{S \in \binom{n}{\leq t}} |\mathcal I_S| \widehat \phi(S_{\underline i})^2
\]
and therefore~\eqref{socp} is equivalent to the following pair of primal/dual second-order cone programs: 
\begin{align} \label{eq:socpdeg}
 \min \quad & 2\epsilon &  \max \quad & -w+\sum_{x \in D} f(x) \phi(x) - \sum_{x \notin D} |\phi(x)|  \\ \notag
\text{s.t.} \quad&  \smash{c = (c_S)_{S \in \binom{n}{\leq t}} \in \R^{\binom{n}{\leq t}}}, \epsilon \in \R  &\text{s.t.} \quad& \phi = (\phi(x))_{x \in \cube} \in \R^{\cube}  \\ \notag
&\Big|\sum_{S \in \binom{n}{\leq t}} c_S \chi_S(x) - f(x)\Big| \le 2\epsilon  \text{ for } x \in D && \sum_{x \in \cube} |\phi(x)| = 1 \\ \notag
&\Big|\sum_{S \in \binom{n}{\leq t}} c_S \chi_S(x) \Big| \le 1+2\epsilon \text{ for } x \notin D && v = \Big(2^n\sqrt{|\mathcal I_S|} \, \widehat \phi(S)\Big)_{S \in \binom{n}{\leq t}} \\ \notag  
&\tilde c = \Big(\frac{c_S}{\sqrt{|\mathcal I_S|}}\Big)_{S \in \binom{n}{\leq t}} && w \geq ||v||_2 \\ \notag
&||\tilde c||_2 \leq 1
\end{align}
We note that strong duality holds since both the primal and dual are strictly feasible. 
\begin{lemma} \label{dualpolycb}
If the optimal value of~\eqref{eq:socpdeg} is strictly larger than $2 \epsilon$, then $\cbdeg(f) > t$. 
\end{lemma}
Hence, the above forms a strengthening of the polynomial method. Indeed, any $\phi$ feasible for the maximization program in~\eqref{eq:approxdeglp} (with objective $>2\epsilon$) will have low-degree Fourier coefficients equal to zero and therefore $(\phi, w=0,v=0)$ will be feasible for the maximization program in~\eqref{eq:socpdeg} (with objective $>2 \epsilon)$. Also, notice that compared to~\eqref{eq:approxdeglp} the primal here has the additional constraint that the coefficients of the approximating polynomial have to be normalized (w.r.t.~a weighted 2-norm).

\begin{center}
    \subsection*{Acknowledgements}
\end{center}
We thank Srinivasan Arunachalam and Jop Bri\"{e}t for bringing this problem to our attention and for many useful discussions. We thank the anonymous referees of QIP'19 for their feedback. 

\bigskip
\newcommand{\etalchar}[1]{$^{#1}$}

\end{document}